\documentclass[letterpaper, 10pt]{article}
\usepackage[a4paper, margin = 1in]{geometry}
\usepackage{xargs}                      
\usepackage[pdftex,dvipsnames]{xcolor}  
\usepackage[utf8]{inputenc}
\usepackage[english]{babel}
\usepackage{amsmath,amssymb,hyperref,xcolor}

\usepackage{amsthm}
\usepackage{enumerate}
\usepackage{subcaption}
\usepackage{refcount}
\usepackage[noadjust]{cite}

\usepackage[colorinlistoftodos,prependcaption,textsize=normalsize]{todonotes}
\newcommandx{\unsure}[2][1=]{\todo[linecolor=red,backgroundcolor=red!25,bordercolor=red,#1]{#2}}
\newcommandx{\change}[2][1=]{\todo[linecolor=blue,backgroundcolor=blue!25,bordercolor=blue,#1]{#2}}
\newcommandx{\info}[2][1=]{\todo[linecolor=OliveGreen,backgroundcolor=OliveGreen!25,bordercolor=OliveGreen,#1]{#2}}
\newcommandx{\improvement}[2][1=]{\todo[linecolor=Plum,backgroundcolor=Plum!25,bordercolor=Plum,#1]{#2}}
\newcommandx{\thiswillnotshow}[2][1=]{\todo[disable,#1]{#2}}
\newcommand{\rom}[1]{\uppercase\expandafter{\romannumeral #1\relax}}

\newcommand{\removed}[1]{{}}

\usepackage{stfloats}

\usepackage{tikz}
\usepackage{pgfplots}
\usepackage{pgfgantt}
\usepackage{pdflscape}
\pgfplotsset{compat=newest}
\pgfplotsset{plot coordinates/math parser=false}

\title{\LARGE \bf Projection-based Controllers with Inherent Dissipativity Properties}

\author{Hoang Chu, S.J.A.M van den Eijnden, W.P.M.H. Heemels   %
\thanks{The authors are with the Control Systems Technology Section, Dept. Mechanical Engineering, Eindhoven University of Technology, The Netherlands. Corresponding author: Hoang Chu
({\tt\small h.chu@tue.nl})
\newline \indent This research received funding form the European Research Council (ERC) under the Advanced ERC grant agreement PROACTHIS, no. 101055384.}%
}

\newtheorem{theorem}{Theorem}

\newtheorem{remark}{Remark}
\newtheorem{assumption}{Assumption}

\theoremstyle{definition}
\newtheorem{definition}{Definition}[section]

\usepackage{accents}

\let\leq\leqslant
\let\geq\geqslant

\let\tilde\widetilde

\let\cal\mathcal

\newcommand{\norm}[1]{\left\lVert#1\right\rVert}    
\newcommand{\remove}[1]{ }
\newcommand{\ree}{\mathbb{R}}

\newcommand{\cS}{{\cal S}}
\newcommand{\cE}{{\cal E}}

\newcommand{\cK}{{\cal K}}

\DeclareMathOperator*{\argmin}{\arg\!\min}

\begin{document}
\maketitle
\thispagestyle{empty}
\pagestyle{empty}

\begin{abstract}
Projection-based Controllers (PBCs) are currently gaining traction in both scientific and engineering communities. In PBCs, the input-output signals of the controller are kept in sector-bounded sets  by means of projection. In this paper, we will show how this projection operation can be used to induce useful passivity or  general dissipativity properties on broad classes of (unprojected) nonlinear controllers that otherwise would not have these properties. The induced dissipativity properties of PBC will be exploited to guarantee asymptotic stability of  negative feedback interconnections of passive nonlinear plants and suitably designed PBC, under mild conditions. Proper generalizations to so-called $(q,s,r)$-dissipativity will be presented as well.  For illustrating the effectiveness of PBC control design via these passivity-based techniques, two numerical examples are provided. 
\end{abstract}

\tikzstyle{block} = [draw, fill=white!20, rectangle, 
    minimum height=2em, minimum width=6em]
\tikzstyle{sum} = [draw, fill=white!20, circle, node distance=1cm]
\tikzstyle{input} = [coordinate]
\tikzstyle{output} = [coordinate]
\tikzstyle{pinstyle} = [pin edge={to-,thin,black}]

\section{Introduction}
Projection-based controllers (PBCs) are gaining attention in several areas of science and engineering, including the control of wafer scanners \cite{Heertjes_2020} and  micro electro-mechanical systems \cite{Shi_2022}. The main philosophy underlying PBC is to keep the input-output pair of the controller confined to prescribed sector-bounded sets. A notable example of PBC is given by the so-called hybrid integrator-gain system (HIGS) \cite{DeeSha_AUT21a}, in which projection is used to force the sign of the integrator's output similar to that of its input at all times, thereby facilitating the possibility to overcome fundamental performance limitations of classical linear time-invariant (LTI) controllers \cite{Eijnden_2020}. The sign equivalence, or more generally the satisfaction of sector-bounds by the inputs and outputs of  controllers is recognized as an important property for performance enhancement, and, is,  for instance also exploited in reset controllers where, in contrast to PBC, this is enforced by resets of the controller states, see, e.g., \cite{NESIC2008,ZaccarianBook}. 

A rigorous foundation for the formalization and analysis of generic closed-loop PBC systems has been provided recently in \cite{DeeSha_AUT21a, Sharif_2019, heemelsaneel_lcss_2023} through the framework of extended projected dynamical systems (ePDS). This new class of discontinuous dynamical systems resembles the classical projected dynamical systems (PDS) (see, e.g., \cite{nagurney1995projected,Dupuis1993}) and also allows to correct the vector field at the boundary of the constraint set. However, in ePDS, next to the consideration of irregular constraint sets, a key difference to classical PDS is the use of  {\em partial} projection operators in which not all projection directions are allowed. In this way,  projections can be captured  that  only correct the controller dynamics and states and do not change the plant dynamics. Existence and completeness results for solutions to ePDS with sector constraints and closed-loop PBC systems (with and without inputs) have recently been established in \cite{heemelsaneel_lcss_2023}.

In the present paper we study (asymptotic) stability properties of generic closed-loop PBC systems. We will start by demonstrating that through the  projection of the input-output pair of a rather general class of (unprojected) single-input single-output controllers (strict) on a well-crafted sector-like set, PBCs can be created with desirable {passivity} or other dissipativity properties in a natural way. Passivity and dissipativity are fundamental properties that have been extensively studied for linear and nonlinear systems and are frequently exploited to facilitate closed-loop system analysis and design for nonlinear systems \cite{khalil2013nonlinear, TORA_passivity,LeTeel_SoftReset}. The celebrated passivity theorem (see, e.g., \cite[Chapter 6]{khalil2013nonlinear}) states that the negative feedback interconnection of two passive systems is passive, and, under  additional detectability assumptions, is asymptotically stable. In view of the aforementioned properties, the notion of passivity and its dissipativity generalizations such as $(q,s,r)$-dissipativity \cite{ANTSAKLIS2013379}, thus provide a natural framework for the analysis and design of PBC systems.  

In line with the above, the main contribution of this paper is to connect the two frameworks of PBC and dissipativity for developing an effective way to design stabilizing controllers for nonlinear plants that satisfy $(q,s,r)$-dissipativity properties. In particular, we will match the sectors of the PBC with the underlying dissipativity properties of the plant in a manner such that the feedback interconnection of the plant and PBC is guaranteed to produce outputs that asymptotically tend to zero. Under  mild additional assumptions on the underlying \emph{unprojected} controller dynamics, stronger asymptotic stability properties can be guaranteed. We care to point out that in our approach considered in this paper, the {unprojected} controller dynamics play no critical role in guaranteeing the closed-loop system to produce outputs that asymptotically tend to zero. This forms an interesting feature of the projection mechanism as it can turn a large class of controllers into PBCs with relevant dissipativity properties in a natural way. Moreover, this suggests that a choice for the unprojected controller dynamics (such as a linear integrator in HIGS \cite{DeeSha_AUT21a}) can be motivated largely from a performance perspective, and thus provides an additional tuning knob for performance that is not evidently present in unprojected controllers. 

The remainder of the paper is structured as follows. Section~\ref{sec:sys config and prob formulation} introduces the problem formulation of the controller design as well as the characteristics of the plant, the projection-based controller, and the closed-loop system. Section~\ref{sec: design} presents the design of PBC controllers that guarantee closed-loop stability. In Section~\ref{sec: analysis} we present our main results in the form of a stability analysis of the closed-loop systems with projection-based controllers. Section~\ref{sec: examples} provides two illustrative examples, and the  conclusions are given in Section~\ref{sec: conclusion}.

\section{System description and problem formulation} \label{sec:sys config and prob formulation}

In this section we will provide the plant that is considered in this paper, along with some of its key properties. We will also motivate the controller configuration and PBC setup, as well as the main problem formulation. 



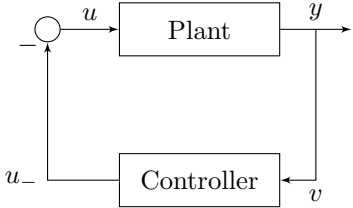
\begin{figure}[tbp]
    \centering
    \begin{tikzpicture}[auto, node distance=2cm,>=latex']
    \node [input, name=input] {};
    \node [sum, right of=input] (sum) {};
    \node [block, right of=sum] (plant) {Plant};
    \node [output, right of=plant] (output) {};
    \node [block, below of=plant] (controller) {Controller};

    \draw [->] (sum) -- node {$u$} (plant);
    \draw [->] (plant) -- node [name=y] {$y$}(output);
    \draw [->] (y) |- node {$v$} (controller);
    \draw [->] (controller) -| node[pos=0.99] {$-$} node {$u_-$} (sum);
\end{tikzpicture}
    \caption{Negative feedback interconnection of a plant and a controller.}
    \label{fig:interconnection}
\end{figure}


\subsection{System configuration and plant model}
In this paper, we will consider the standard negative feedback interconnection of plant and controller as shown in Fig. \ref{fig:interconnection}. The  plant is assumed to be a single-input-single-output (SISO) nonlinear system of the form
\begin{subequations} \label{eq:plant5}
\begin{eqnarray}
\dot{x} & = & f_p(x, u) \\ \label{eq:plant output}
y &= & G_px
\end{eqnarray} 
\end{subequations}
with  plant state $x\in \ree^n$, control input $u\in \ree$ and output $y\in \ree$. Here, $f_p:\ree^{n}\times \ree \rightarrow \ree^n$ is a continuous function and $G_p\in \ree^{n}$ an output matrix describing a linear mapping of plant states to plant output. In this paper we specifically consider plants of the form \eqref{eq:plant5} that are \emph{detectable}, and possess certain \emph{dissipativity} properties. To make the upcoming discussions precise, we adopt the following definitions from \cite{Isidori1999NonlinearCS,LeTeel_SoftReset,ANTSAKLIS2013379}.

\begin{definition}(\cite[Definition 10.7.3]{Isidori1999NonlinearCS})
The plant \eqref{eq:plant5} is said to be detectable if for any initial condition $x_0$, the solution $x(t)$ to $\dot x = f_p(x,0)$ with $x(0) = x_0$ is defined for all $t \geq 0$, and
$$y = 0,\, \forall t \geq 0 \implies \lim_{t \to \infty} x(t) = 0$$
\end{definition}

\begin{definition} \label{def: qsr}
(\cite[Definition 2]{ANTSAKLIS2013379})
Consider $q,s,r \in \ree$.
    The system \eqref{eq:plant5} is $(q,s,r)$-dissipative, if there exist a continuously differentiable storage function $V_p: \ree^n \rightarrow \ree_{\geq 0}$ and $\mathcal{K}_\infty$-functions $\alpha_1$ and $\alpha_2$, such that for all $x\in \ree^n$
    \begin{equation}
        \alpha_1(\|x\|)\leq V_p(x) \leq \alpha_2(\|x\|)
    \end{equation}
    and 
    for all $(x,u) \in \ree^n \times \ree$ it holds that
        \begin{align}\label{eq:dV}
    \frac{\partial V_p}{\partial x} f_p(x,u)\leq q u^2 + 2suy + r y^2 = \begin{pmatrix}
             u \\ y
         \end{pmatrix}^\top \underbrace{\begin{pmatrix}
             q &s\\s &r
         \end{pmatrix}}_{:= M}\begin{pmatrix}
             u \\ y
         \end{pmatrix}.
    \end{align}
\end{definition}



The notion of $(q,s,r)$-dissipativity generalizes several important properties including passivity by taking $q = r =0$ and $s >0$. Note that we can scale the storage function $V_p$ in this case to get $s=\frac{1}{2}$ as is more common for passivity giving $\frac{\partial V_p}{\partial x} f_p(x,u)\leq uy$. Also, output-strictly passivity  \cite{khalil2013nonlinear}, which complies with $q  =0$, $r<0$, and $s = \frac{1}{2}$ (or $s>0$ and apply scaling to the storage function again) and ${\cal L}_2$-type properties by having $s=0$, $q>0$ and $r<0$ (implying an ${\cal L}_2$-gain  smaller than $\sqrt{\frac{q}{-r}}$ for the plant \eqref{eq:plant5}).  

Summarizing, we state the following standing assumption on the characteristics of the plant \eqref{eq:plant5} that will be used throughout the rest of the paper.

\begin{assumption}\label{assumption: detectable}
    The plant \eqref{eq:plant5} is detectable and $(q,s,r)$-dissipative.
\end{assumption}

The main objective in this paper is to design a controller, in particular a projection-based controller, that asymptotically stabilizes the plant \eqref{eq:plant5} by a systematic and natural design framework exploiting the dissipativity properties of the plant dynamics. 



\subsection{Projection-based control}

In order to introduce PBC, we start with the unprojected (SISO) controller dynamics that are  described by 
\begin{subequations}\label{eq:controller5}
\begin{eqnarray}
\dot{z} &= & \frac{d}{dt}\begin{pmatrix}
    z_1\\z_2
\end{pmatrix} = \begin{pmatrix}
    f_1(z_1,z_2,y)\\ f_2(z_1,z_2,v)
\end{pmatrix} = f_c(z,v)\\
u_- &= & z_1
\end{eqnarray} 
\end{subequations}
with the controller state $z=(z_1,z_2)\in \ree \times \ree^{m-1}$, controller output $u_-\in \ree$ and controller input $v\in \ree$. The map $f_c: \ree^{m}\times \ree \rightarrow \ree^m$ is assumed to be continuous. Note that the output equation  in \eqref{eq:controller5}, and also in \eqref{eq:plant5},  is chosen in a  (particular) linear form,  which, for many dynamical systems can be realized by an appropriate state transformation, e.g.,   into the normal form \cite{khalil2013nonlinear}. The reason for this selected structure is the ease of exposition and analysis in the next sections. 

Using the descriptions for the plant  \eqref{eq:plant5} and the controller \eqref{eq:controller5}, in the negative feedback interconnection of Fig.~\ref{fig:interconnection}, where $v=y$ and $u=-u_-$, we obtain the (unprojected) closed-loop dynamics as 
\begin{equation} \label{eq:unprojected}
\dot \xi =    f(\xi) = (f_p(x, -z_1),f_c(z,G_p x)), 
\end{equation} where the closed-loop states are $\xi=(x,z) \in \ree^{n+m}$.

To introduce the PBC based on \eqref{eq:controller5},  we will add a \emph{projection operation} to \eqref{eq:unprojected}, with the objective to keep $(v,u_{-})$ in a well-designed sector $
S_{[k_1,k_2]}$ described by
\begin{equation} \label{eq:sector5}
 S_{[k_1,k_2]} :=   \left\{(v,u_{-})\in \mathbb{R}^2 \mid  (u_{-}-k_1v)(u_{-}-k_2v) \leq 0 \right\}, 
\end{equation}
where $k_1, k_2 \in \ree$ with $k_1<k_2$, as is motivated by HIGS and other PBC systems \cite{DeeSha_AUT21a}. If $k_1$, $k_2$ are clear from the context, we write $S_{[k_1,k_2]}=S$.
Following \cite{heemelsaneel_lcss_2023}, we directly describe the dynamics of the closed-loop system, by defining 
\begin{equation}
 {\cal S}=   \left\{\xi\in \mathbb{R}^{n+m} \mid {(G_px, z_1)}\in S \right\} 
\end{equation}
and the projection subspace 
\begin{equation}
    \cE := \textup{Im}{\begin{bmatrix} O_{n\times m} \\ I_m \end{bmatrix}},
\end{equation} 
chosen in such that a way only changes in the controller dynamics are induced by the projection, as we will explain below.

We introduce now the partial projection operator $\Pi_{{\cal S},\cE}$  in \eqref{eq:unprojected}, which leads to the following closed-loop description in the form of an extended projected dynamical system  (ePDS)  \cite{heemelsaneel_lcss_2023, DeeSha_AUT21a}
\begin{equation}\label{eq:cloop}
    \dot \xi = F(\xi)= \Pi_{{\cal S},\cE}(\xi,f(\xi))
\end{equation}
where the  projection operator is given by 
\begin{equation}\Pi_{\cS,\cE}(\xi,p) := \argmin_{w \in T_{\cS}(\xi), w - p \in \cE} \|w-p\|. \label{piSE}
\end{equation}
Here, $T_{\cS}(\xi)$ is the tangent cone to the set $\cS \subset \mathbb{R}^n$ at a point $\xi \in \cS$,  defined as the collection of all vectors $p\in \mathbb{R}^n$ for which there exist sequences $\{x_i\}_{i\in \mathbb{N}} \in S$ and $\{\tau_i\}_{i\in \mathbb{N}}$, $\tau_i > 0$
with $x_i \rightarrow x$, $\tau_i \downarrow 0$ and $i \rightarrow \infty$, such that $
    p = \lim_{i\rightarrow \infty} \frac{x_i - x}{\tau_i}.$  The tangent cone $T_{\cS}(\xi)$ at $\xi$ captures, loosely speaking, the admissible velocites at the point $\xi \in \cS$ that keep the flow or the solutions inside $\cS$, which is exactly the objective. With this interpretation in mind,  the key idea behind $\Pi_{S,\cE}(\xi,f(\xi))$ is that it basically ``projects'' the unprojected vector field $f(\xi)$ on $T_{\cS}(\xi)$, i.e., it finds the  velocity in $T_\cS(\xi)$ that is closest to $f(\xi)$) to make sure that the solutions \eqref{eq:cloop} do not leave the set $\cS$, much as in the classical projected dynamical systems (PDS) \cite{nagurney1995projected}. However, note that, in contrast to classical PDS, here irregular constraint sets are used, namely sector sector-like sets, and that the projection is ``partial'' in the sense that the projection direction $w - f(\xi)$ has to be contained in $\cE$. Given the definition of  $\cE$ in our case, this indicates that only the controller dynamics can be modified by the projection, and the plant dynamics, which adhere to physical laws, cannot be changed. Clearly, this is in line with the controller setting that we consider in this paper, as the plant dynamics are fixed. Note that also that when $\xi$ is in the interior of $\cS$ the tangent cone $T_{\cS}(\xi)=\ree^{m+n}$ and, hence, $\dot \xi =f(\xi)$, i.e., in the interior of the constraint set the original dynamics are active.

\subsection{Problem formulation} \label{subsec:pf}

The problem to be addressed in this paper is to develop a systematic and easy-to-apply design framework for PBC systems for the  asymptotic stabilization of plants having $(q,s,r)$-dissipativity properties in the feedback interconnection of Fig. \ref{fig:interconnection}. Hence, in particular, the design question is how to choose the controller dynamics 
\eqref{eq:controller5} together with the  sector $S_{[k_1,k_2]}$, i.e., the selection of $k_1<k_2$, such that the interconnection \eqref{eq:cloop} is globally asymptotically stable, in the sense of the following definition. 

\begin{definition} \label{def: asymp stab} (\cite[Definition 4.1]{khalil2013nonlinear})
    We say that the system \eqref{eq:cloop} is globally asymptotically stable, if 
    \begin{enumerate}[(a)]
    \item for any  $\xi_0 \in \cS$, there is a solution \footnote{Solutions are in the sense of Carath\'eodory as formalized in Definition~\ref{def.cara} below.} $\xi:\ree_{\geq 0} \rightarrow \ree^{m+n}$ to \eqref{eq:cloop} with $\xi(0)=\xi_0$ that is defined for all times $t\in \ree_{\geq 0}$, and all solutions can be prolonged to be defined for all times $t\in \ree_{\geq 0}$,
    \item for all $\varepsilon>0$ there is a $\delta >0$ such that for all $\xi_0 \in \cS$ with $\|\xi_0\| \leq \delta$ all corresponding solutions with $\xi(0)=\xi_0$ satisfy $\|\xi(t)\| \leq \delta$  for all $t\in \ree_{\geq 0}$,
    \item for all solutions $\xi:\ree_{\geq 0} \rightarrow \ree^{m+n}$ to \eqref{eq:cloop} it holds that 
    $\lim_{t \to \infty} \xi(t) = 0.$
    \end{enumerate}
\end{definition}

As we will see in the next sections, the (unprojected) controller dynamics \eqref{eq:controller5} play only a minor role in the stabilization;  the design of the sector $S_{[k1,k2]}$ is the main step in the design and under mild assumptions on \eqref{eq:controller5}  the closed-loop system becomes globally asymptotically stable. Interestingly, this hints upon the fact that existing smooth (possibly non-stabilizing) controllers can be easily turned into stabilizing controllers by just adding suitable projection operators (aligned with the $(q,s,r)$-dissipativity properties of the plant) to the controller.

\section{Design of projection-based controllers}\label{sec: design}

In this section we present the systematic design of the PBC, which given a plant that is detectable and $(q,s,r)$-dissipative as in Assumption \ref{assumption: detectable}, such that  the resulting closed-loop system 
\eqref{eq:cloop} is asymptotically stable. The design procedure involves two important steps: i) design of the sector in \eqref{eq:sector5} from projection, and ii) design of the underlying (unprojected) controller dynamics. First, the sector in \eqref{eq:sector5} is designed to guarantee that the storage function of the plant is strictly decreasing if $(y,-u) \in S$ and $(y,-u) \neq (0,0)$, i.e., we need to ensure that the input-output pair of the controller (respectively the input-output pair of the plant) satisfies $qu^2 + 2suy + ry^2 < 0$ for all $(y,-u) \in S$ and $(y,-u) \neq (0,0)$. For this purpose, we can make use of the sector inequality in \eqref{eq:sector5} in an S-procedure relaxation manner \cite{Petersen2000} to construct a matrix inequality that can be solved in order to find $k_1$ and $k_2$. Specifically, if there exists a solution $\lambda \geq 0$ to the following matrix inequality
\begin{equation} \label{eq: matrix ineq}
    \begin{bmatrix}
        q & s\\ s & r
    \end{bmatrix} - \lambda \begin{bmatrix}
        1 & \frac{1}{2} (k_1+k_2) \\
        \frac{1}{2}(k_1+k_2) & k_1k_2
    \end{bmatrix} \prec 0,
\end{equation}
then we can guarantee that $qu^2 + 2suy + ry^2 < 0$ for all $(u,y) \in \{ (u,y)\in \ree^2 \mid (u+k_1 y)(u+k_2 y) \leq 0 \}$.
In case the first matrix in \eqref{eq: matrix ineq} (corresponding to the matrix $M$ in \eqref{eq:dV}) is positive definite, \eqref{eq: matrix ineq} will not admit a feasible solution for any $\lambda \geq 0$, $k_1,k_2 \in \mathbb{R}$. Moreover, if this matrix is negative definite, there will always exist a feasible solution to \eqref{eq: matrix ineq} - take $\lambda =0$ for example. As such, the interesting case to consider is for $q\geq 0$ ($q \leq 0$) and $qr-s^2 \leq 0$ ($qr-s^2 \geq 0$). It is easy to see that in this case $\lambda >0$ is a necessary condition for feasibility of \eqref{eq: matrix ineq}. As such, we can scale the above matrix inequality by a factor of $\bar{\lambda}:= \frac{1}{\lambda}$. Moreover, using a change of variables $c:=k_1+k_2$ and $d:= k_1 k_2$ results in a linear matrix inequality (LMI) that can be easily solved using numerical tools. We present the design of the sector that ensure that \eqref{eq: matrix ineq} is feasible in three interesting cases.

\begin{enumerate}
\item Passive plant ($q = 0$, $r=0$ and $s>0$): 
\begin{equation}
   0 < k_1 < k_2 < \infty.
    \end{equation}
    \item Output-strictly passive plant ($q = 0$, $r<0$, and $s > 0$): 
\begin{equation}
    \frac{r}{2s}< k_1 < k_2 < \infty.
\end{equation} 
    In this case, $k_1$ may be smaller than 0 as $r<0$ and $s>0$.
    \item General dissipative plant ($q>0$ and $qr-s^2 <0$): 
\begin{equation}
    \frac{s - \sqrt{s^2-qr}}{q} < k_1 <    k_2 < \frac{s + \sqrt{s^2-qr}}{q}.
\end{equation}
This case includes other cases of passivity such as input-strict passivity (\cite{khalil2013nonlinear}) and other cases of dissipativity such as the bounded $\mathcal{L}_2$-gain case ($q>0,s=0,r<0$).
\end{enumerate}

In general, a well-designed sector satisfies the following assumption.
\begin{assumption}\label{assumption: sector}
Two scalars $k_1$ and $k_2$ that define the sector $S$ as in \eqref{eq:sector5} satisfy \eqref{eq: matrix ineq} for some $\lambda \geq 0$.
\end{assumption}
If Assumption \ref{assumption: sector} is satisfied, and thus the inequality in \eqref{eq: matrix ineq} is solved successfully, due to the feedback interconnection we can guarantee
\begin{equation}\label{eq: Vp dot}
    \dot{V}_p(x(t)) \leq 0, \quad \textrm{and} \quad \dot{V}_p(x(t)) = 0 \Leftrightarrow (u,y) = (0,0).
\end{equation}

The second aspect of the design of a PBC controller is the design of the underlying (unprojected) controller dynamics. We need only one condition on the controller dynamics to ensure that the closed-loop dynamics are asymptotically stable, namely that the $z_2$-dynamics of the unprojected controller in \eqref{eq:controller5} is input-to-state stable (ISS) with respect to $(z_1,v)$. 
\begin{definition}
\label{assumption: iss}(\cite[Definition 4.7]{khalil2013nonlinear})
    The system $\dot{z}_2 = f_2(z_1,z_2,v)$ is input-to-state stable with respect to $(z_1,v)$, if there exist a  $\mathcal{K}\mathcal{L}$-function $\bar{\beta}$ and $\mathcal{K}$-functions $\bar{\gamma}_1$ and $\bar{\gamma}_2$ such that for all essentially bounded measurable functions $z_1$ and $v$ and all  $t\in \ree_{\geq 0}$ it holds that 
\begin{equation} \label{eq: ISS z_2}
    \| z_2(t) \| \leq \bar{\beta}(\| z_2(0)\|,t) + \Bar{\gamma}_1(\|z_1\|_\infty) + \Bar{\gamma}_2(\|v\|_\infty)
\end{equation}
with $\|v\|_\infty := \text{ess} \sup_{\tau \in \ree_{\geq 0}} \| v(\tau) \|$, and  $\|z_1\|_\infty := \text{ess} \sup_{\tau \in \ree_{\geq 0}} \| z_1(\tau) \|$.     
\end{definition}

In the remainder we assume the $z_2$-dynamics to be ISS in the sense of Definition~\ref{assumption: iss}. The main reason for requiring this condition becomes clear later in Section~\ref{sec: analysis}.

\begin{remark}
Note that this ISS condition is different from the minimum-phase property. Loosely speaking, a system is called minimum phase if the zero dynamics is asymptotically stable. Hence, in the case of \eqref{eq:controller5}, the controller is minimum phase, if $\dot z_2 = f_2(z_1,z_2,y)$ is asymptotically stable provided that $u(t)=z_1(t) = 0$ for all times $t$ and thus $\dot z_1(t) = f_1(0,z_2(t),y(t)) = 0$. However, for the PBC, $\dot z_1(t) = 0$ may hold even when $f_1(0,z_2(t),y(t)) \neq 0$, because the projection operator alters the controller dynamics. Hence, a stronger condition in terms of ISS of the $z_2$-dynamics is needed than non-minimum-phase-ness of the controller.
\end{remark}

\section{Main results}\label{sec: analysis}
 
Equipped with the above results and definitions, we are now ready to state the main results of the paper. We will start with a discussion on the existence of solutions followed by proving the asymptotic stability of the closed-loop system under the stated assumptions and proposed design.

\subsection{Existence of solutions} \label{sec:existence}
Before studying any (stability or passivity) properties of \eqref{eq:cloop}, let us first focus on the existence of Carath\'eodory-type solutions to the closed-loop dynamics given an initial condition. 

\begin{definition} \label{def.cara}
A function $\xi:[0,T]\rightarrow \ree^{n}$ is said to be a (Carath\'eodory) solution to $\dot \xi = \Pi_{\cS, \cE}(\xi,f(\xi))$, if $\xi$ is absolutely continuous on $[0,T]$ and satisfies $\dot \xi(t) = \Pi_{\cS,\cE}(\xi(t),f(\xi(t)))$ for almost all $t\in[0,T]$ and $\xi(t)\in \cS$ for all $t\in[0,T]$. 
Furthermore, $\xi:[0,\infty)\rightarrow \ree^{n}$ is called a solution, if the restriction of $x$ to $[0,T]$ is a solution on $[0,T]$ for each $T>0$.
\end{definition}
In \cite{heemelsaneel_lcss_2023}, it has been demonstrated that for any initial condition $\xi_0\in \ree^n$ there is a $T>0$ such that a Carathéodory solution $\xi:[0,T]\rightarrow \ree^{n}$ exist for the considered dynamics \eqref{eq:cloop} with $\xi(0)=\xi_0$. 

We recall a few ingredients from the proof of existence of solutions given in \cite{heemelsaneel_lcss_2023} as they turn out to be important for the closed-loop stability analysis below.  The proof in \cite{heemelsaneel_lcss_2023} starts by showing the existence of so-called Krasovskii solutions, i.e., (locally) absolutely functions satisfying, almost everywhere, the Krasovskii regularization of \eqref{eq:cloop}, which is given by the differential inclusion
\begin{equation} \label{eq:kras}
    \dot \xi \in \cap_{\delta >0} \overline{\textup{con}}F(B(\xi,\delta)) =: K_F(\xi).
\end{equation}
Here,  $\overline{\textup{con}}(M)$ denotes the closed convex hull of the set $M$, in other words, the smallest closed convex set containing $M$. Then, for the case of  sectors $\cS = \cK \cup -\cK$ as considered here, it is shown than any Krasovskii solution  
is, in fact, a Carath\'eodory solution to \eqref{eq:cloop}, thereby showing local existence of Carath\'eodory solution given an initial condition.



\subsection{Stability analysis} \label{sec: stability analysis}

In this section, we provide the main result of the paper, which pertains to asymptotic stability of the closed-loop system \eqref{eq:cloop}. The asymptotic stability to be proven is in the sense of Definition \ref{def: asymp stab}, i.e., the solution to the closed-loop dynamics \eqref{eq:cloop} is bounded and converges to zero. 

We summarize the main result of the paper in the following theorem.
\begin{theorem} 
Consider the plant \eqref{eq:plant5} and suppose Assumption \ref{assumption: detectable} is satisfied. Furthermore, suppose the PBC controller \eqref{eq:controller5} to have a sector condition satisfying Assumption \ref{assumption: sector} and the $z_2$-dynamics to be ISS with respect to $(z_1,v)$. 

Then the closed-loop PBC system \eqref{eq:cloop} is globally asymptotically stable.
\end{theorem}
\begin{proof}
The global asymptotic stability of \eqref{eq:cloop} will be proved by showing that the input and output of the plant converge to zero, then the state of the PBC and the plant also converge to zero due to the ISS property and detectability, respectively.  

Consider $\xi(0)\in \cS$ and the corresponding solution $\xi:[0,T]\rightarrow \ree^{m+n}$ to \eqref{eq:cloop}. The local existence of solutions is proven in \cite{heemelsaneel_lcss_2023} and we proceed to show that all trajectories are bounded. Note that at this point, we cannot conclude directly that $T=\infty$, however, by ruling out finite-escape times due to boundedness of the solution, we will establish $T=\infty$ in this proof as well. Based on the discussion in Section~\ref{sec:existence}, it is clear that a solution to \eqref{eq:cloop} is also a solution to its Krasovskii regularization \eqref{eq:kras}, which will turn out useful in the application of the invariance principle for differential inclusions as in \cite{RyanInvariancePrinciple}, which is used to prove that all trajectories converge to the invariant set that contains only the origin, implying that the origin is asymptotically stable (Definition \ref{def: asymp stab}).

First, we show the boundedness of solutions on the time interval $t\in [0,T]$. Because of the designed sector condition and the negative feedback interconnection, the time derivative of the storage function $t \mapsto V_p(x(t))$ is non-positive (recall \eqref{eq: Vp dot}). Hence,  $V_p(x(t))$ is non-increasing over time, leading to
\begin{equation}
    \alpha_1(\|x(t)\|) \leq V_p(x(t)) \leq V_p(x(0)) \leq \alpha_2(\|x(0)\|),
\end{equation}
and thus 
\begin{equation} \label{eq: bound x}
    \|x(t)\|\leq \alpha_1^{-1}(\alpha_2(\|x(0))\|),\, \forall t\in [0,T].
\end{equation} 
From \eqref{eq: bound x} and \eqref{eq:plant output}, we get the existence of a $\mathcal{K}$-function $\beta$ such that
\begin{equation}
    \|y(t) \| \leq \beta(\|x(0) \|),\,\forall t\in [0,T].
\end{equation}
Using now the sector condition $(y,-u) \in S$ as in \eqref{eq:sector5}, we get for all $t\in [0,T]$
\begin{equation}
    \| z_1(t)\| = \|u(t) \|  \leq \max(|k_1|,|k_2|) \| y(t) \| \leq \kappa \beta(\|x(0) \|)
\end{equation}
with $\kappa:=\max(|k_1|,|k_2|)$. 
Because the projection operator does not alter the $z_2$-part of the controller dynamics (see \cite[Sec. \rom{3}.B]{heemelsaneel_lcss_2023} for more details), $z_2$ satisfies $\dot{z}_2 = f_2(z_1,z_2,y)$ and thus the ISS property \eqref{eq: ISS z_2} can be exploited. This gives 
\begin{equation}
    \| z_2(t) \|\leq  \bar{\beta}(\| z_2(0)\|,0) + \Bar{\gamma}_1 (\beta(\norm{x(0)})) + \Bar{\gamma}_2 (\kappa\beta(\norm{x(0)}))
\end{equation}
The states $z_2$ are therefore bounded, implying the boundedness of the solutions $\xi$ given an inital condition $\xi(0)=\xi_0\in\cS$ on $[0,T]$, and, in fact, the Lyapunov stability part (b) of Definition~\ref{def: asymp stab}. As $\xi$ is absolutely continuous, and thus uniformly continuous (on compact intervals), combined with the boundedness of $\xi$, we can  now proceed similarly as in the proof of \cite[Thm. 4.2]{DeeSha_AUT21a} to show by contradiction that $T$ can be taken as $T=\infty$. Moreover, every solution can be prolonged to be defined on $\ree_{\geq 0}$, and thus part (a) of Definition~\ref{def: asymp stab} is guaranteed. Furthermore, in the terminology of \cite{RyanInvariancePrinciple} this means satisfaction of the completeness  and the precompactness property \cite[Definition 2.3]{RyanInvariancePrinciple} of solutions, that is needed to apply the generalized invariance principle in \cite[Theorem 2.11]{RyanInvariancePrinciple} to find the largest invariant set.

As the generalized invariance principle in \cite[Theorem 2.11]{RyanInvariancePrinciple} applies to differential inclusions with certain regularity properties, we consider the trajectories $\xi:[0,\infty) \rightarrow \ree^{m+n}$ now as solutions to the Krasovskii regularization. The map $\xi \mapsto K_F(\xi)$ is outer semicontinuous \cite[Lemma 5.16]{goebel_sanfelice_teel_2012} and has non-empty, convex, and compact values \cite{heemelsaneel_lcss_2023}. Since $K_F$ is also locally bounded, it follows that $K_F$ is upper semicontinuous, see \cite[Lemma 5.15]{goebel_sanfelice_teel_2012}. From the above we get that $V(\xi):= V_p(x)$ is a locally Lipschitz continuous (in fact, continuously differentiable), positive semi-definite function, and  
$$\begin{aligned}
    V^o(\xi,\phi)&:= \limsup_{y\to \xi,h\downarrow 0} \frac{V(y+h\phi)-V(y)}{h} \\
    &= \frac{\partial V_p(x)}{\partial x} f_p(x,-z_1)
    \leq 0
\end{aligned}$$ Here, we used the fact that the projection does not affect the dynamics of the plant. 
According to \cite[Theorem 2.11]{RyanInvariancePrinciple}, for some constant $c\geq 0$, $\xi(t)$ approaches the largest weakly invariant set \cite[Definition 2.7]{RyanInvariancePrinciple} in $\Sigma \cap V^{-1}(c)$ when $t \to \infty$, where
\[\Sigma = \{ \xi \mid \tilde{u}(\xi) \geq 0 \}\]
and $\tilde{u}(\xi) =  \dot V_p =\frac{\partial V_p(x)}{\partial x} f_p(x,-z_1)$. As $ \tilde{u}(\xi)\leq 0$, $\Sigma = \{ \xi \mid \tilde{u}(\xi) = 0 \} = \{ \xi \mid \dot V_p = 0 \}$ = $\{ \xi \mid (u,y) = (0,0) \}$ (from \eqref{eq: Vp dot}). Hence, any solution inside (the largest weakly invariant set of) $\Sigma \cap V^{-1}(c)$ satisfies $u(t) =y(t) =0$ for all times $t$, and thus also $v(t)=z_1(t)=0$. Because both $z_1$ and $v$ are zero for all times, the ISS property of the $z_2$-dynamics  with respect to $(z_1,v)$, guarantees that $z_2(t) \to 0$ as $t\to \infty$. 
Furthermore, because of the detectability property (Assumption \ref{assumption: detectable}) 
of the plant $x(t) \to 0$ for $t\to \infty$ as $u(t) = 0$ and $y(t) = 0$ for all times $t$. Thus, all solutions converge to  the origin $\xi = 0$. Hence, also part (c) of Definition~\ref{def: asymp stab} is satisfied, and, thus, the closed-loop system \eqref{eq:cloop} is globally asymptotically stable. This completes the proof.

\end{proof}

Let $I(\xi)=\{ i \in \{1,2 \} \,|\, H_i (\xi) \xi = 0 \}$ denote the set of active constraints at $\xi \in \cS$. According to \cite{BardiaThesis}, the solution $\eta^\star$ to the optimization problem satisfies
\begin{equation*}
\begin{aligned}
    \eta^\star &= (E^\top E)^{-1} E^\top H^\top_{I(\xi)} \lambda_{I(\xi)}\\ &= \pm 
        I_{m}
     \begin{pmatrix}
        I_{m} &0_{m \times 1}
    \end{pmatrix} \begin{pmatrix}
        1 &0 &\cdots &0 &(\cdot)
    \end{pmatrix}^\top \lambda_{I(\xi)}\\ &= \pm\begin{pmatrix}
        1 &0 &\cdots &0 
    \end{pmatrix}^\top \lambda_{I(\xi)},
\end{aligned}
\end{equation*}
with $ \lambda_{I(\xi)} \in \ree$, thus the projected dynamics of the controller is only altered in the $z_1$-dimension. The $\pm$ sign depends on the state $\xi$, or more specifically on $v$, and does not affect this analysis. In other words, the dynamics of $z_2$ remain unaffected by the projection operator.

In this section we illustrate the PBC design framework on two numerical examples. 

\subsection{Mass-spring-damper system}
The control problem is to stabilize a mass-spring-damper system given the velocity only. The dynamics of the mass-spring-damper is
\begin{equation*}
\begin{aligned}
    \frac{d}{dt} \begin{pmatrix}
        x_1 \\ x_2
    \end{pmatrix} &= \begin{pmatrix}
        0 &1\\ -10 &-0.01
    \end{pmatrix} \begin{pmatrix}
        x_1 \\ x_2
    \end{pmatrix} + \begin{pmatrix}
        0 \\ 1
    \end{pmatrix} u,\\
    y &= \begin{pmatrix}
        0 &1
    \end{pmatrix} \begin{pmatrix}
        x_1 \\ x_2
    \end{pmatrix}.
\end{aligned}
\end{equation*}
With a storage function that is identical to the total energy, the above system mass-spring-damper system is passive from $u$ to $y$.
A simple choice is a gain controller $u = -ky$ ($k>0$), and with this choice, there is a limitation in the convergence rate \cite{TORA_passivity}. Note that the mass-spring-damper system is clearly detectable as $u(t)=y(t)=0, \, \forall t \geq 0 \implies x(t) = 0$. We compare this simple controller with $k=4.8$ (C0) and a projected controller (C1) that has the unprojected dynamics
\begin{equation} \label{eq: mck controller}
\begin{aligned}
    \frac{d}{dt} \begin{pmatrix}
        z_1\\z_2
    \end{pmatrix} &= \begin{pmatrix}
        1 &-10\\0 &-1
    \end{pmatrix}\begin{pmatrix}
        z_1\\z_2
    \end{pmatrix} + \begin{pmatrix}
        0\\1
    \end{pmatrix}y_2,\\
    u &= -z_1,
\end{aligned}
\end{equation}
and a sector condition $(k_1,k_2) = (3.5,6.0)$. Both controllers are stabilizing controllers (Fig. \ref{fig:mck compare}). 
The closed-loop system with a projected controller has faster convergence and lower overshoot, compared with the closed-loop system with C0 (which has a damping ratio of $\zeta = 0.77$ - a favorable value for the balance between fast convergence and low overshoot). It is notable that the linear controller with the dynamics \eqref{eq: mck controller} (and without any projection) is a destabilizing controller, yet it can be turned into a stabilizing controller only by applying the projection operator to enforce a sector condition (Fig. \ref{fig:mck compare}). 

\begin{figure}[tbp]
    \centering
    \includegraphics[width=0.5\textwidth]{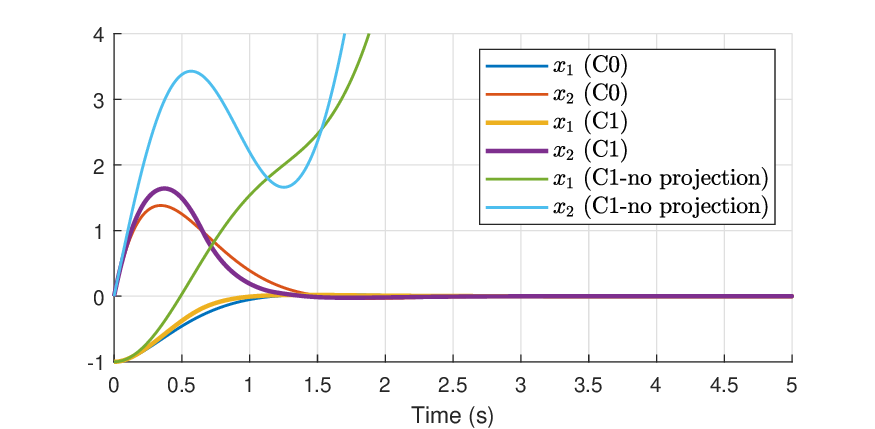}
    \caption{Comparison simple gain controller and projected controller on a mass-spring-damper system.}
    \label{fig:mck compare}
\end{figure}

\begin{figure}[tbp]
    \centering
    \includegraphics[width=0.5\textwidth]{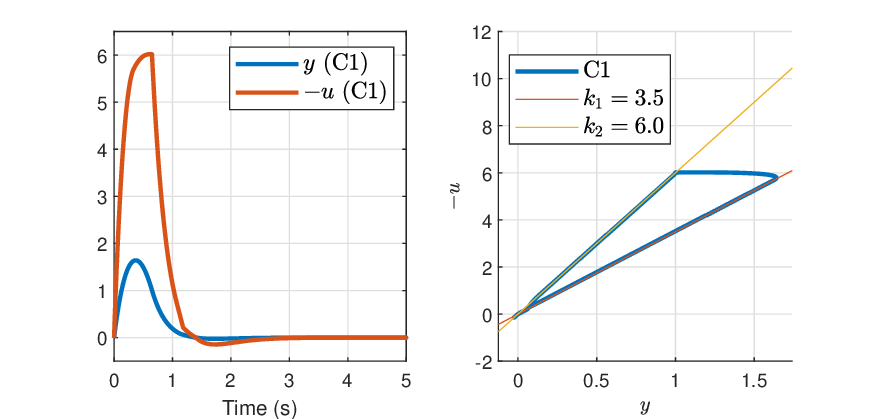}
    \caption{Projected controller input-output operating on a mass-spring-damper system.}
    \label{fig:mck inout}
\end{figure}

\subsection{TORA system}
The control problem to be considered is to design an auxiliary input to a passivating controller that stabilizes the TORA system - a well-known nonlinear, underactuated, and passive system. The TORA system, comprising a spring, a mass, and a pendulum attached to the mass, is studied in \cite{khalil2013nonlinear,LeTeel_SoftReset,TORA_passivity}, and several controllers that exploit the passivity of the system were proposed. The TORA system is described by the following differential equations
\begin{equation}
    \begin{pmatrix}
        1 &\epsilon \cos (\theta) \\ \epsilon \cos (\theta) &1
    \end{pmatrix} \begin{pmatrix}
        \Ddot{\theta} \\ \Ddot{x}
    \end{pmatrix} = \begin{pmatrix}
        u \\ -x + \epsilon\Dot{\theta}^2 \sin (\theta)
    \end{pmatrix},
\end{equation}
where $\epsilon = 0.1$, $\theta$ is the angle of the pendulum, $x$ is the normalized position of the mass, and $u$ is the input to the system.
Following \cite{TORA_passivity}, by considering the storage function
\begin{equation}
    W = \frac{h_0 + 1}{2} [(z_1-\epsilon \sin y_1)^2+z_2^2] + \frac{h_1}{2} y_1^2 + \frac{1}{2} y_2^2 (1-\epsilon^2 \cos^2 y_1),
\end{equation}
with $h_0,h_1>0$ and the transformation of variables
\begin{equation*}
    \begin{aligned}
        z_1 &= x + \epsilon \sin \theta,\\
        z_2 &= \dot x + \epsilon \dot \theta \cos \theta,\\
        y_1 &= \theta,\\
        y_2 &= \dot \theta,
    \end{aligned}
\end{equation*}
one finds a passivating input 
\begin{equation}\label{eq: passivating control}
    u = -h_0 \epsilon \cos y_1 (-z_1 + \epsilon \sin y_1) -h_1 y_1 + w
\end{equation}
that renders the system passive from $u$ to $y_2$. Therefore, a simple choice to stabilize the system is $w = -h_2 y_2$ ($h_2 >0$). The parameters found for high performance are $(h_0,h_1,h_2) = (10,1,0.5)$ \cite{TORA_passivity}, and this controller is referred to as C0. We compare this control law with two versions of our projected controller. One version (C1) has the original dynamics
\begin{equation}
\begin{aligned}
    \frac{d}{dt} \begin{pmatrix}
        z_1\\z_2
    \end{pmatrix} &= \begin{pmatrix}
        3 &-2\\0 &-3
    \end{pmatrix}\begin{pmatrix}
        z_1\\z_2
    \end{pmatrix} + \begin{pmatrix}
        0\\1
    \end{pmatrix}y_2,\\
    w &= -z_1,
\end{aligned}
\end{equation}
and a sector condition $(k_1,k_2) = (0.45,0.6)$. The $z_2$-dynamics is clearly ISS as $-3$ is a Hurwitz matrix of dimension one. The other version (C2) has the original nonlinear dynamics
\begin{equation}
\begin{aligned}
    \frac{d}{dt} \begin{pmatrix}
        z_1\\z_2
    \end{pmatrix} &= \begin{pmatrix}
        3 &-2\\0 &-1
    \end{pmatrix}\begin{pmatrix}
        z_1\\z_2
    \end{pmatrix} + \begin{pmatrix}
        0\\-2 z_2^3 + (1+z_2^2)y_2^2
    \end{pmatrix},\\
    w &= -z_1,
\end{aligned}
\end{equation}
and the same sector condition $(k_1,k_2) = (0.45,0.6)$. The $z_2$-part of the controller dynamics is ISS \cite[Example 4.26]{khalil2013nonlinear}. One proves that the TORA system with the passivating control \eqref{eq: passivating control} is detectable using LaSalle's invariance principle similarly to \cite{TORA_passivity}. 
The stability analysis in Sec. \ref{sec: stability analysis} therefore applies, and the origin of the closed-loop system is asymptotically stable (Fig. \ref{fig:tora compare}, Fig. \ref{fig:tora controller}). 
The projected controller provides a similar convergence time with less oscillation, in both $x$ and $\theta$ (Fig. \ref{fig:tora compare}). It is to be emphasized that both controllers C1 and C2 would be destabilizing if they were without the projection operator. The design of these PBC controllers provides a simple way to turn destabilizing controllers into stabilizing ones while giving a high level of freedom to the design of the underlying un-projected dynamics.

\begin{figure}[tbp]
    \centering
    \includegraphics[width=0.5\textwidth]{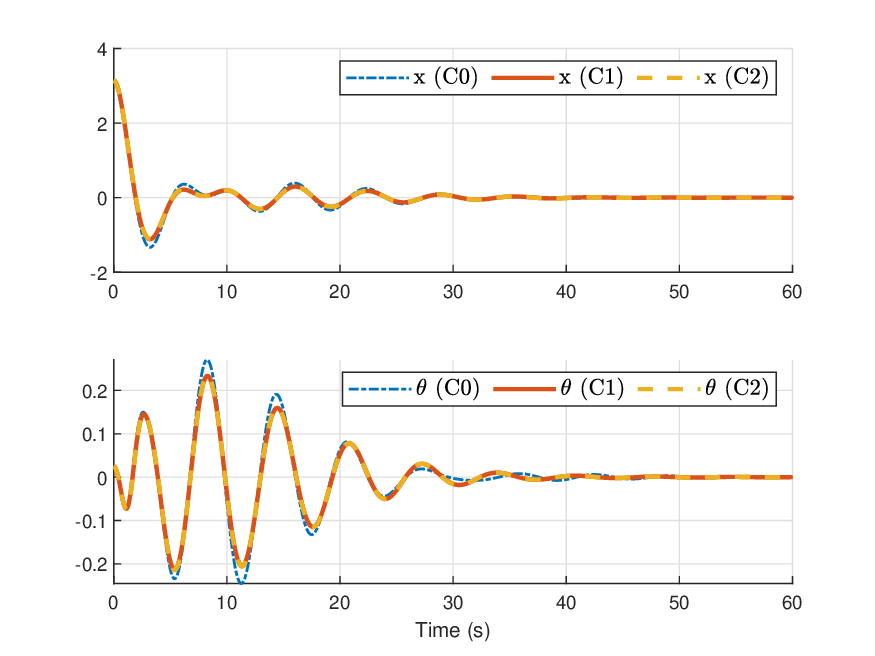}
    \caption{Time evolution of the position of the mass and the angle of the pendulum of the TORA system with three controllers (C0: simple gain, C1 and C2: PBC)}
    \label{fig:tora compare}
\end{figure}

\begin{figure}[tbp]
    \centering
    \includegraphics[width=0.5\textwidth]{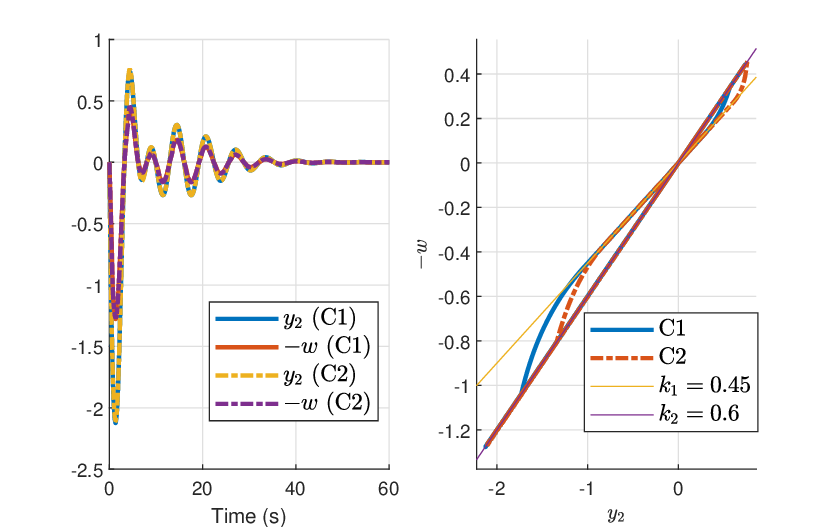}
    \caption{The input-output pair of two PBC controllers operating on the TORA system (plotted against time and against each other)}
    \label{fig:tora inout}
\end{figure}

\begin{figure}[tbp]
    \centering
    \includegraphics[width=0.5\textwidth]{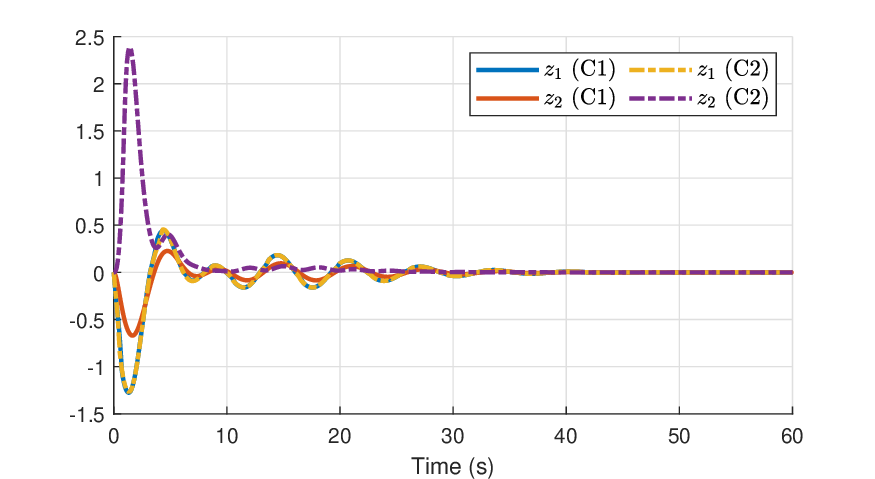}
    \caption{Time evolution of the states of two PBC controllers}
    \label{fig:tora controller}
\end{figure}

\section{Conclusion}\label{sec: conclusion}
In this paper, we present a systematic design procedure for  projection-based controllers that asymptotically stabilize a detectable plant satisfying a general $(q,s,r)$-dissipativity property (including, e.g., passivity, strict passivity, or ${\cal L}_2$-gains). The design is intuitive and consists mainly of shaping the sector-bounds imposed on the input-output relationship of the PBC by means of partial projection operators. In fact, any given controller dynamics, in which an input-to-state stability property holds on the part of the controller dynamics that is unchanged by the projection, can be turned into an asymptotically stabilizing controller by the insertion of a well-crafted projection in the control loop. The proof of the main result relies on a generalized invariance principle for differential inclusions. Two illustrative examples are provided to demonstrate the application of PBC design to stabilization of well-known passive systems. 

\bibliographystyle{IEEEtran}
\bibliography{ref}
\end{document}